\documentclass[journal]{IEEEtran}
\ifCLASSINFOpdf
\else
\fi
\usepackage{caption}
\usepackage{cite}
\usepackage{color}
\usepackage{amsmath}
\usepackage{enumerate}
\usepackage{enumitem}
\usepackage{mathrsfs}
\usepackage{makeidx}         
\usepackage{graphicx}        
\usepackage{multicol}        
\usepackage{subfigure}
\usepackage{epsfig,amssymb,latexsym}
\usepackage{psfrag}
\usepackage{float}
\usepackage{stfloats}
\usepackage{epsfig,amssymb,latexsym}
\usepackage{psfrag}
\usepackage[ruled,vlined]{algorithm2e}
\usepackage{threeparttable}    
\usepackage{algorithm2e}
\usepackage{caption}
\usepackage{subfigure}
\usepackage{multirow}
\usepackage{amsthm}
\usepackage{booktabs}

\newtheorem{assumption}{Assumption}
\newtheorem{remark}{Remark}
\newtheorem{theorem}{Theorem}
\newtheorem{lemma}{Lemma}

\newcommand{\bi}{\begin{itemize}}
	\newcommand{\ei}{\end{itemize}}
\newcommand{\be}{\begin{equation}}
	\newcommand{\ee}{\end{equation}}
\newcommand{\bd}{\begin{displaymath}}
	\newcommand{\ed}{\end{displaymath}}
\newcommand{\bea}{\begin{eqnarray}}
	\newcommand{\eea}{\end{eqnarray}}
\newcommand{\ba}{\begin{array}}
	\newcommand{\ea}{\end{array}}
\newcommand{\bc}{\begin{center}}
	\newcommand{\ec}{\end{center}}

\usepackage{graphicx}
\usepackage{booktabs}

\begin{document}

\title{Selective Memory Recursive Least Squares: {Recast Forgetting into Memory in RBF Neural Network Based Real-Time Learning}}
%
%
%

\author{Yiming Fei,~\IEEEmembership{Student Member,~IEEE}, Jiangang Li$^{*}$,~\IEEEmembership{Senior Member,~IEEE}, Yanan Li$^{*}$,~\IEEEmembership{Senior Member,~IEEE}

\thanks{Y. Fei is with the School of Mechanical Engineering and Automation, Harbin Institute of Technology, Shenzhen 518055, China.}
\thanks{*Correspondence: J. Li is with the School of Mechanical Engineering and Automation, Harbin Institute of Technology, Shenzhen 518055, China Email: jiangang\_lee@hit.edu.cn}
\thanks{*Correspondence: Y. Li is with the Department of Engineering and Design, University of Sussex, Brighton BN1 9RH, UK Email: yl557@sussex.ac.uk}
}




\maketitle

\begin{abstract}
{In radial basis function neural network (RBFNN) based real-time learning tasks, forgetting mechanisms are widely used such that the neural network can keep its sensitivity to new data.} However, with forgetting mechanisms, some useful knowledge will get lost simply because they are learned a long time ago, which we refer to as the passive knowledge forgetting phenomenon. To address this problem, this paper proposes a real-time training method named selective memory recursive least squares (SMRLS) in which the classical forgetting mechanisms are recast into a memory mechanism. {Different from the forgetting mechanism, which mainly evaluates the importance of samples according to the time when samples are collected, the memory mechanism evaluates the importance of samples through both temporal and spatial distribution of samples}. With SMRLS, the input space of the RBFNN is evenly divided into a finite number of partitions and a synthesized objective function is developed using synthesized samples from each partition. {In addition to the current approximation error, the neural network also updates its weights according to the recorded data from the partition being visited. Compared with classical training methods including the forgetting factor recursive least squares (FFRLS) and stochastic gradient descent (SGD) methods, SMRLS achieves improved learning speed and generalization capability, which are demonstrated by corresponding simulation results.} 

\end{abstract}



\IEEEpeerreviewmaketitle

\section{Introduction}
\IEEEPARstart{A}s a result of its linearly parameterized form and universal approximation capabilities, the radial basis function neural network (RBFNN) has become one of the most popular function approximators in the field of adaptive control and system identification \cite{r1,r2,r3,r4,r5}. Since the parameter convergence of approximators can improve the control performance and robustness of adaptive systems, the real-time learning problem of the RBFNN has received extensive attention \cite{r6,r7,r8,r9}. Among different real-time training methods, the stochastic gradient descent (SGD) based and recursive least squares (RLS) based methods have been widely applied to various systems \cite{r10, r11, r12, r13}. Since the Lyapunov stability theory based analysis often leads to a SGD based weight update law, SGD has become one of the most popular training methods in the field of neural network control \cite{r14,r15,r16}. Compared with SGD, the RLS method achieves higher convergence speed by adjusting the adaptation gain matrix \cite{r17,r18}. However, the adaptation gain of ordinary RLS method will continually decrease as new samples are learned such that the neural network gradually becomes insensitive to new data, degrading the effectiveness of RLS under parameter perturbation. 

{In order to keep the neural network sensitive to new data, various forgetting mechanisms have been proposed to improve the RLS method. As the most common forgetting mechanism, the forgetting factor makes the neural network pay more attention to new samples by reducing the weight of past samples in the objective function \cite{r19}. To further improve the performance of the forgetting factor recursive least squares (FFRLS) method, different variable forgetting factors have been proposed. Specifically, a gradient-based variable forgetting factor is proposed to improve the tracking capability of RLS, in which the value of the forgetting factor is adjusted by the gradient of the mean square error \cite{r20}. In \cite{r21}, an improved variable forgetting factor recursive least squares (VFFRLS) method is proposed by adjusting the forgetting factor according to the time-averaging estimate of the autocorrelation of a priori and a posteriori errors. In \cite{r22}, a typical variable forgetting factor is proposed to summarize a class of variable-rate-forgetting methods, and the convergence and consistency of the algorithm are also analyzed. In addition to forgetting factors, some other forgetting mechanisms including the sliding window forgetting, covariance resetting and covariance modification can also keep the neural network sensitive to new data \cite{r23,r24,r25,r26}. }

{However, with these forgetting mechanisms, the improved sensitivity to new data is obtained at the cost of forgetting the past data. Since past training samples may contain valuable information about the system, forgetting of them will prevent the accumulation of knowledge in the neural network. In this paper, a phenomenon named passive knowledge forgetting is introduced to represent the decreasing weight of past samples in the objective function over time. Corresponding analysis shows that this phenomenon happens in both RLS methods with forgetting mechanisms and the SGD method. To suppress this phenomenon, an algorithm that evaluates the importance of samples not only based on their temporal distribution is required to be developed.}

In this paper, an algorithm named selective memory recursive least squares (SMRLS) is proposed to simultaneously suppress passive knowledge forgetting and keep the RBFNN sensitive to new data. With SMRLS, the forgetting mechanisms are recast into a memory mechanism which evaluates the importance of samples through both temporal and spatial distribution of samples. The implementation of SMRLS follows the following steps: 
\begin{enumerate}
	\item{The input space of the RBFNN is normalized and evenly divided into finite partitions.  
	}
	\item{Memory update functions are designed according to a priori knowledge to synthesize samples within the same partition into one sample. 
	}
	\item{A synthesized objective function is proposed by summing up the squared approximation errors of the synthesized samples and optimized through a recursive algorithm. 
	}
\end{enumerate}

While the distribution of partitions represents the evaluation criteria for the importance of samples on the spatial scale, the design of memory update functions reflects the criteria on the temporal scale. Based on the assumption that the latest sample within a partition has higher reliability than past samples within the same partition, a specific memory update function is designed such that the synthesized sample within a partition is set to the latest sample collected from the same partition. Compared with SGD and RLS with forgetting mechanisms, SMRLS achieves the following merits: 
\begin{enumerate}
	\item{While making the neural network sensitive to new data, SMRLS suppresses the passive knowledge forgetting phenomenon and thus obtains faster learning speed.  
	}
	\item{{The generalization capability of the learned knowledge between different tasks is improved because the non-uniformity between the training set and testing set is suppressed by the approximately uniform distribution of the synthesized samples.}  
	}
\end{enumerate}

{It should be noted that a method named variable-direction forgetting recursive least squares (VDFRLS) also has the potential to suppress the passive knowledge forgetting phenomenon by restricting forgetting to the information-rich subspace \cite{r27,r28,r29,r30,r31,r32}. However, it is computationally complex compared with SMRLS. Moreover, there is still no explicit conclusion about the objective function corresponding to its update law, making it difficult to predict the performance of this algorithm before specific tasks are carried out \cite{r32}. }

The rest of this paper is organized as follows: Section \ref{section2} introduces some preliminaries and formulates the real-time function approximation problem. In Section \ref{section3}, the implementation and theoretical analysis of SMRLS are presented in detail. Section \ref{section4} discusses the computational complexity, learning performance, prospects and possible improvements of SMRLS. Corresponding simulation results are given in Section \ref{section5} to demonstrate the effectiveness of SMRLS. Finally, Section \ref{section6} concludes the main results of this paper.

\section{Problem Formulation And Preliminaries}\label{section2}
In Section \ref{section2-1}, an open-loop real-time function approximation problem is formulated. Section \ref{section2-2}, \ref{section2-3}, \ref{section2-4} provide a brief introduction to the RBFNN based function approximation, SGD method and RLS method, respectively.

\subsection{Real-Time Function Approximation: An Open-Loop Case}\label{section2-1}
Consider an open-loop real-time approximation task where the output of the approximator will not influence the state trajectory of the system. The function to be approximated is described as follows:   
\begin{equation}\label{eq1}
	\left\{ \begin{gathered}
		y = f(x)+v \hfill \\
		x = g(t) \hfill \\ 
	\end{gathered}  \right.
\end{equation}
where $g(t):\left[ {0,\infty } \right) \mapsto {\mathbb{R}^n}$, $f(x):{\mathbb{R}^n} \mapsto \mathbb{R}$ are unknown continuous functions, $x \in {\Omega _x} \subseteq {\mathbb{R}^n}$ is an independent variable defined on the compact set $\Omega_x$, $y \in \mathbb{R}$ is the measurement of $f(x)$, and $v\in \mathbb{R}$ is the measurement noise. 

\begin{assumption}\label{assumption1}
To simplify the subsequent derivation, it is assumed that the measurement is accurate such that $v=0$. 
\end{assumption}

Assume that the output of the approximator is $\hat f(x,\theta)$ with the parameter $\theta \in \mathbb{R}^q$. To obatin an approximation which has good generalization capability over $\Omega _x$, the following objective function is considered: 
\begin{equation}\label{eq2}
J(\theta ) = {\int_{{\Omega _x}} {\left\| {f(x) - \hat f(x,\theta )} \right\|} ^2}dx. 
\end{equation}
The optimal parameter $\theta^*$ is defined as:
\begin{equation}\label{eq3}
{\theta ^ * }  \triangleq  \mathop {\arg \min }\limits_{\theta  \in {\mathbb{R}^q}} J(\theta ). 
\end{equation}

Let $i = 1,2, \ldots ,m$ denote the serial number of the sampling time, $m$ denote the total number of samples, and the real-time approximation objective is to constantly improve the estimation of $\theta^*$ as samples $(x(i),y(i))$ keep being collected. To estimate $\theta^*$ with sampled data, consider the following sum of squared error (SSE) objective function: 
\begin{equation}\label{eq4}
{J_f}(\theta ) =  \sum\limits_{i = 1}^m {{{\left\| {y(i) - \hat f(x(i),\theta )} \right\|}^2}}. 
\end{equation}
If the samples $(x(i),y(i))$ are uniformly distributed over the compact set $\Omega_x$ under Assumption \ref{assumption1}, we obtain: 
\begin{equation}\label{eq5}
\mathop {\lim }\limits_{m \to \infty } \frac{1}{m} {J_f}(\theta ) = J(\theta )
\end{equation}
which indicates the importance of the uniform distribution of samples for accurate approximation \cite{r33}. Let $\theta _f^*$ denote the minimizer of ${J_f}(\theta )$ as follows: 
\begin{equation}\label{eq6}
\theta _f^* \triangleq \mathop {\arg \min }\limits_{\theta  \in {\mathbb{R}^q}} {J_f}(\theta ). 
\end{equation}
A sufficient condition for $\theta _f^*$ to estimate $\theta ^*$ accurately can be expressed as: i) the number of samples is sufficient, and ii) the samples are uniformly distributed over $\Omega_x$. For real-time approximation tasks, the condition is difficult to satisfy because of the difficulties in controlling the distribution of samples. 

\subsection{RBFNN Based Function Approximation}\label{section2-2}
The RBFNN is a class of single hidden layer feedforward neural networks whose activation function is the radial basis function. As a result of its linearly parameterized form, low computational complexity and fast convergence speed, it has been applied to various fields including adaptive control, system identification and pattern recognition \cite{r2,r3,r4,r5,r34,r35}. The output of an RBFNN can be formulated as:  
\begin{equation}\label{eq7}
{f_{NN}}(\chi ) = \sum\limits_{i = 1}^N {{w_i}{\phi _i}(\chi ) = } {W^T}\Phi (\chi )
\end{equation}
where $\chi  \in {\mathbb{R}^l} $ is the input vector, ${f_{NN}}(\chi ) \in \mathbb{R}$ is the output, $W = {[{w_1},{w_2}, \ldots ,{w_N}]^T} \in {\mathbb{R}^N}$ is the weight vector, $N$ is the number of neurons, and $\Phi (\chi ) = {\left[ {{\phi _1}(\chi ),{\phi _2}(\chi ), \ldots ,{\phi _N}(\chi )} \right]^T} \in {\mathbb{R}^N}$ is the regressor vector composed of radial basis functions ${\phi _i}(\chi ),i = 1,2, \ldots ,N$. Without loss of generality, ${\phi _i}(\chi )$ is set as the Gaussian function in this paper: 
\begin{equation}\label{eq8}
{\phi _i}(\chi ) = \exp \left( { - {{\left\| {\chi  - {c_i}} \right\|}^2}/2\sigma _i^2} \right),i = 1,2, \ldots N
\end{equation}
where ${c_i} \in {\mathbb{R}^l}$ is the center and ${\sigma _i} \in \mathbb{R}$ is the receptive field width of ${\phi _i}(\chi )$. 

Consider the function approximation problem in Section \ref{section2-1}. RBFNNs have the ability to approximate any continuous function $f(x):{\mathbb{R}^n}  \mapsto  \mathbb{R}$ over a compact set ${\Omega _x } \subseteq {\mathbb{R}^n}$ to arbitrary accuracy on the premise of sufficient neurons and suitably designed centers and receptive field widths \cite{r1}. Then $f(x)$ can be approximated as: 
\begin{equation}\label{eq9}
f(x ) = {W^{*T}}\Phi (x ) + \epsilon(x ),\forall x  \in {\Omega _x }
\end{equation}
where $W^{*}\in {\mathbb{R}^N}$ is the optimal weight vector of the approximation, and $\epsilon (x)$ is the bounded approximation error satisfying $\left| {\epsilon(x )} \right| < {\epsilon^*},\forall x \in {\Omega _x }$. In this paper, the optimal value of $W^{*}\in {\mathbb{R}^N}$ is defined as follows: 
\begin{equation}\label{eq10}
{W^*} \triangleq \mathop {\arg \min }\limits_{W \in {\mathbb{R}^N}} \left\{ {{{\int_{{\Omega _x }} {\left\| {f(x ) - {W^T}\Phi (x )} \right\|}^2 }}dx } \right\}
\end{equation}	
where the compact set ${\Omega _x }$ is known as the input space of the RBFNN. 

{According to different settings of neuron centers and receptive field widths, the RBFNN can be divided into localized networks and global networks \cite{r7,r36}. For global RBFNN, the approximation over a region of the input space is carried out by all of the neurons. By contrast, the localized RBFNN accomplishes the local approximation only using the neurons around corresponding region of the input space. }

\subsection{Stochastic Gradient Descent Based Training for RBFNNs}\label{section2-3}
The gradient descent (GD) method is one of the most effective methods to train the RBFNN \cite{r37,r38}. When samples are collected in chronological order, the batch size of GD is set to one such that the SGD method is obtained and used to train the RBFNN in real-time. To estimate $W^*$, the following squared error objective function is considered: 
\begin{equation}\label{eq11}
{J_{GD}}(W,k) = \frac{1}{2}{\left\| {y(k) - {W^T}\Phi (x(k))} \right\|^2}
\end{equation}
where $k$ is the current sampling time. The gradient of ${J_{GD}}(W,k)$ corresponding to $W$ is formulated as: 
\begin{equation}\label{eq12}
	\begin{aligned}
\nabla {J_{GD}}(W,k) &= \frac{{\partial {J_{GD}}(W,k)}}{{\partial W}} \\
 &=  - \Phi (x(k))\left[ {y(k) - {W^T}\Phi (x(k))} \right]. 
 	\end{aligned}
\end{equation}
Then $W$ is updated in the direction of $-\nabla {J_{GD}}(W,k)$ as follows: 
\begin{equation}\label{eq13}
W(k + 1) = W(k) - \eta \nabla {J_{GD}}(W,k)
\end{equation}
where $\eta  \in {\mathbb{R}^{+}}$ is a designed learning rate. 

\subsection{Recursive Least Squares Based Training for RBFNNs}\label{section2-4}
Thanks to its linearly parameterized form, the RBFNN can learn with linear optimization methods such as RLS \cite{r2}. The classical RLS method with a forgetting factor is formulated as follows: 
\begin{itemize}[leftmargin=*]
	\item \textbf{RBFNN-RLS}
	\begin{equation}\label{eq14}
		\begin{small}
			\begin{aligned}
				&{W}(k) = {W }(k - 1) + P(k)\Phi ({{  x(k)}})\left[ {{{  y(k)}} - W ^T(k - 1)\Phi ({{  x(k)}})} \right] \hfill \\
				&{P^{ - 1}}(k) = \lambda{P^{ - 1}}(k - 1) + \Phi ({{  x(k)}}){\Phi ^T}({{  x(k)}}) \hfill \\
				&\Phi ({{  x(k)}}) = {\left[ {{\phi _1}({{  x(k)}}),{\phi _2}({{  x(k)}}), \ldots ,{\phi _N}({{  x(k)}})} \right]^T} \hfill \\ 
			\end{aligned} 
		\end{small}
	\end{equation}
where $P(k)$ is the covariance matrix at sampling time $k$ and $\lambda  \in \left( {0,1} \right]$ is the forgetting factor. {With matrix inversion lemma, the calculation of $P^{-1}(k)$ is avoided and \eqref{eq14} is reformulated as follows \cite{r39}.}
\item \textbf{RBFNN-RLS (matrix inversion lemma)}
\begin{equation}\label{eq15}
	\begin{small}
		\begin{aligned}
			&{W }(k) = {W }(k - 1) + P(k)\Phi ({{  x(k)}})\left[ {{{  y(k)}} - W ^T(k - 1)\Phi ({{  x(k)}})} \right] \hfill \\
			&P(k) = \frac{1}{\lambda }\left[ {P(k - 1) - \frac{{P(k - 1)\Phi (x(k)){\Phi ^T}(x(k))P(k - 1)}}{{\lambda  + {\Phi ^T}(x(k))P(k - 1)\Phi (x(k))}}} \right] \hfill \\
			&\Phi ({{  x(k)}}) = {\left[ {{\phi _1}({{  x(k)}}),{\phi _2}({{  x(k)}}), \ldots ,{\phi _N}({{  x(k)}})} \right]^T}. \hfill \\ 
		\end{aligned} 
	\end{small}
\end{equation}
\end{itemize}

\begin{lemma}[Objective Function of RLS \cite{r19}]\label{lemma1}
Let $W_0 \in {\mathbb{R}^N}$, a positive definite matrix $P_0 \in {\mathbb{R}^{N \times N}}$ be the initial values of $W(k)$ and $P(k)$, respectively. For all $k>0$, denote the minimizer of the function
\begin{equation}\label{eq16}
		\begin{aligned}
{J_{RLS}}(W,k) =& \sum\limits_{i = 1}^k {{\lambda ^{k - i}}{{\left\| {y(i) - {W^T}\Phi (x(i))} \right\|}^2}}  \\
&+ {\lambda ^k}{\left( {W - {W_0}} \right)^T}P_0^{ - 1}\left( {W - {W_0}} \right)
		\end{aligned}
\end{equation}
by
\begin{equation}\label{eq17}
W(k) \triangleq \mathop {\arg \min }\limits_{W \in {\mathbb{R}^N}} {J_{RLS}}(W,k).
\end{equation}
Then $W(k)$ is given by \eqref{eq14} or \eqref{eq15} for all $k>0$. 
\end{lemma}

\begin{remark}[Passive Knowledge Forgetting]\label{remark1}
In an optimization process, the phenomenon that the weight of each sample in the objective function continually decreases over time is named passive knowledge forgetting in this paper. As shown in \eqref{eq11}, \eqref{eq16}, both the SGD and FFRLS methods are accompanied with this phenomenon. Since the learned input-output mapping from the past samples will be gradually forgot by the neural network, passive knowledge forgetting has negative effects on knowledge accumulation in the learning process. 
\end{remark}

\section{Selective Memory Recursive Least Squares}\label{section3}
In this section, a real-time training method for the RBFNN is designed to satisfy the following requirements at the same time: i) sensitivity to new data, ii) suppression of passive knowledge forgetting, iii) approximately uniform distribution of samples considered by the objective function. 
\subsection{Design of the Synthesized Objective Function}\label{section3-1}
{To obtain the objective function of SMRLS, the input space $\Omega_x$ of the RBFNN is normalized into a compact set $\Omega$ which is evenly divided into $N_P$ lattice partitions ${\Omega ^j},j = 1,2, \ldots ,{N_P}$. The partition ${\Omega ^j}$ which has been sampled at least once in the real-time learning process will generate a synthesized sample $({\gamma _j}(k),{\varphi _j}(k))$ with ${\gamma _j}(k) \in {\Omega ^j}$ being the synthesized input and ${\varphi _j}(k)$ being the synthesized output according to all samples within the partition. When the RBFNN is used to approximate $f(x)$ in Section \ref{section2-1}, we have the synthesized input ${\gamma _j}(k) \in {\Omega ^j} \subseteq {\mathbb{R}^n}$, and the synthesized output ${\varphi _j}(k) \in \mathbb{R}$. }

\begin{assumption}\label{assumption2}
The initial value of $({\gamma _j}(k),{\varphi _j}(k))$ is set to ${\gamma _j}(0) = {0_{n \times 1}},{\varphi _j}(0) = 0$ for $j = 1,2, \ldots {N_P}$. 
\end{assumption}

To demonstrate the normalization and partitioning explicitly, an example where $x(k) = {\left[ {{x_1}(k),{x_2}(k)} \right]^T}$ is given as follows: 
\begin{enumerate}
\item{As shown in Fig. \ref{fig1}, the input space is normalized into the scope of $\left[ { - 1,1} \right] \times \left[ { - 1,1} \right]$. In this case, the original input is normalized as follows: 
\begin{equation}\label{eq18}
{{\bar x}_i} = \frac{{2{x_i} - \mathop {\max }\limits_{x \in {\Omega _x}} \left\{ {{x_i}} \right\} - \mathop {\min }\limits_{x \in {\Omega _x}} \left\{ {{x_i}} \right\}}}{{\mathop {\max }\limits_{x \in {\Omega _x}} \left\{ {{x_i}} \right\} - \mathop {\min }\limits_{x \in {\Omega _x}} \left\{ {{x_i}} \right\}}},i = 1,2
\end{equation}
where $\bar x = {\left[ {{{\bar x}_1},{{\bar x}_2}} \right]^T} \in {{ \Omega }}$ is the normalized input. }
\item{The normalized input space ${\Omega} $ is evenly divided into $N_P$ lattice partitions ${\Omega ^j},j = 1,2, \ldots ,{N_P}$ which are mutually disjoint. For each partition ${\Omega ^j}$, there is a synthesized sample $({\gamma _j}(k),{\varphi _j}(k))$ generated according to all samples within the partition. 
}
\end{enumerate}
\begin{figure}[htbp]
	\centering
	\includegraphics[scale=0.33]{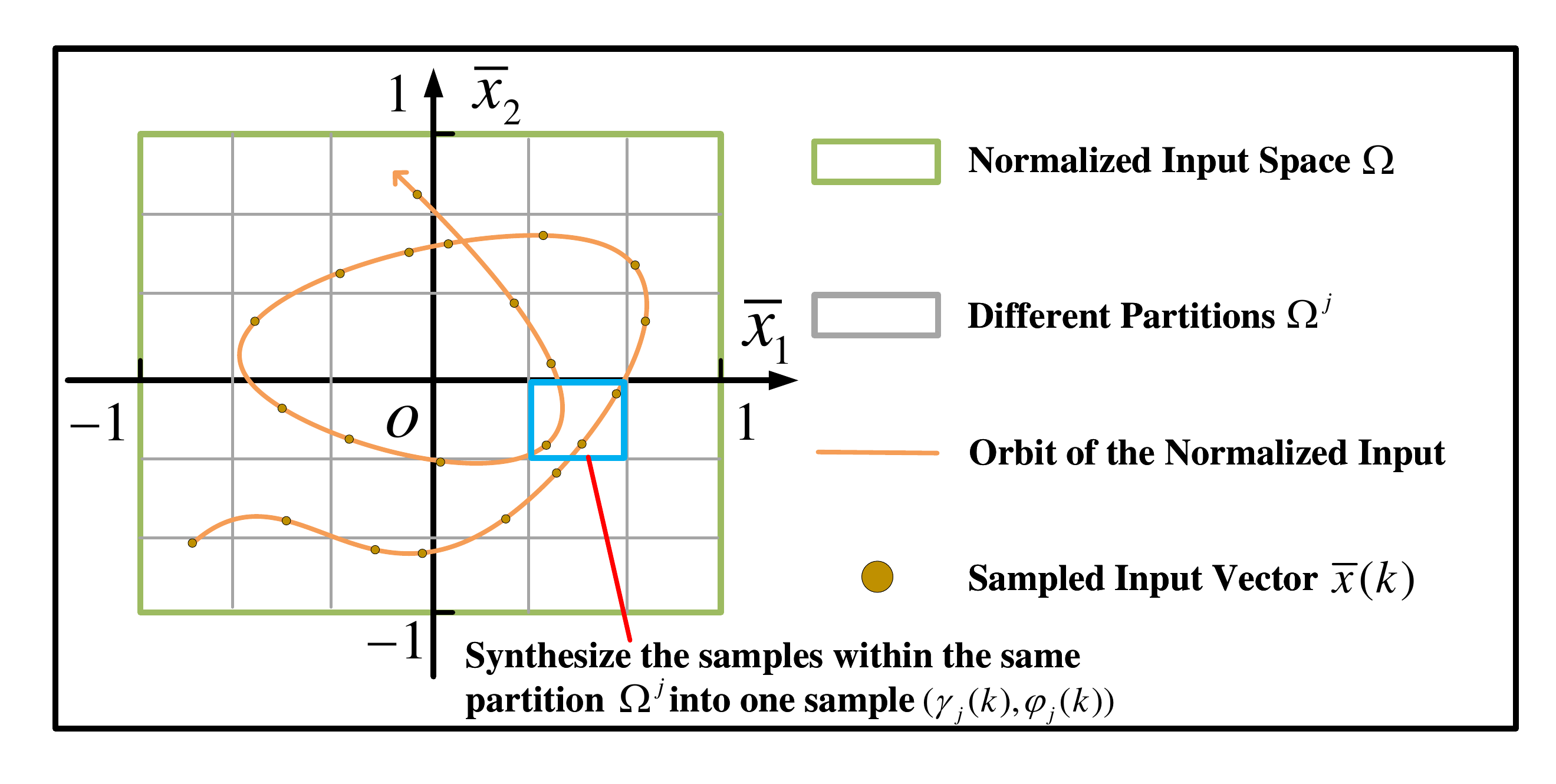}
	\caption{Normalization and partitioning of the input space}\vspace{-1em}\label{fig1}
\end{figure}

\begin{remark}\label{remark2}
The reason for the normalization is to simplify the neuron center distribution of the RBFNN such that $f(x)$ along different state trajectories can be approximated with the same neuron distribution covering the normalized input space. For regularity of the formula, the symbol $x(k)$ is directly used to represent the normalized sampled input at sampling time $k$. 
\end{remark}

Let ${M_P}(k)$ denote the number of partitions which have been sampled at least once before sampling time $k$ (including $k$). Assume that the superscript $j$ of $\Omega^j$ represents that $\Omega^j$ is the $j$th partition visited by $x$ and sampled in the learning process. Similar to \eqref{eq16} with $\lambda=1$, the synthesized objective function at sampling time $k$ is defined as follows: 
\begin{equation}\label{eq19}
	\begin{aligned}
{J_S}(W,k) =& \sum\limits_{j = 1}^{{M_P}(k)} {{{\left\| {{\varphi _j}(k) - {W^T}\Phi ({\gamma _j}(k))} \right\|}^2}}  \\
&+ {\left( {W - {W_0}} \right)^T}P_0^{ - 1}\left( {W - {W_0}} \right)
	\end{aligned}
\end{equation}
where ${P_0} \in {\mathbb{R}^{N \times N}}$ is a designed positive definite matrix. 

\begin{remark}\label{remark3}
	If the state trajectory of $x$ traverses almost all of the $N_P$ partitions of $\Omega$, $({\gamma _j}(k),{\varphi _j}(k))$ are approximately uniformly distributed over $\Omega$ because ${\gamma _j}(k) \in {\Omega ^j} $. 
\end{remark}

\subsection{Design of Memory Update Functions}\label{section3-2}
The synthesized sample $({\gamma _j}(k),{\varphi _j}(k))$ is obtained by synthesizing the samples within ${\Omega ^j}$ according to memory update functions. Consider the following partition judgment function: 
\begin{equation}\label{eq20}
	G(j, x(k)) = \left\{ \begin{gathered}
		0,{\text{   }}for{\text{ }}  x(k) \notin \Omega^j \hfill \\
		1,{\text{   }}for{\text{ }}  x(k) \in \Omega^j  \hfill \\ 
	\end{gathered}  \right.
\end{equation}
which is used to determine in which partition the sample $(x(k),y(k))$ locates. For any $k>1$, $({\gamma _j}(k),{\varphi _j}(k))$ is calculated with the following recursive update laws: 
\begin{equation}\label{eq21}
	\left\{ \begin{aligned}
		&{\gamma _j}(k) = {\gamma _j}(k-1) + G(j, x(k)){S_\gamma }({\gamma _j}(k-1), x(k)) \hfill \\
		&{\varphi _j}(k) = {\varphi _j}(k-1) + G(j, x(k)){S_\varphi }({\varphi _j}(k-1), y(k)) \hfill \\ 
	\end{aligned}  \right.
\end{equation}
where ${S_\gamma }( \cdot )$ and ${S_\varphi }( \cdot )$ are memory update functions. According to \eqref{eq20}, only when the new sample $(x(k),y(k))$ is collected from the partition ${\Omega ^j}$, the synthesized sample $({\gamma _j}(k),{\varphi _j}(k))$ will update in the next time step. 

While the distribution of partitions represents the evaluation criteria for the importance of samples on the spatial scale, the design of memory update functions reflects the criteria on the temporal scale. Since parameter perturbation is common in the field of adaptive control and system identification, the following assumption is considered. 
\begin{assumption}\label{assumption3}
	The latest sample within each partition $\Omega^j,j = 1,2, \ldots ,N_P$ has higher reliability compared with the past samples within the same partition. 
\end{assumption}

With Assumption \ref{assumption3}, the memory update functions are designed as follows: 
\begin{equation}\label{eq22}
	\left\{ \begin{aligned}
		&{S_\gamma }({\gamma _j}(k-1), x(k)) =  - {\gamma _j}(k-1) +  x(k) \hfill \\
		&{S_\varphi }({\varphi _j}(k-1), y(k)) =  - {\varphi _j}(k-1) +  y(k)  \hfill \\ 
	\end{aligned}  \right.
\end{equation}
which means that the synthesized sample of each partition is set to the latest sample within the partition.

\subsection{Recursive Optimization for Synthesized Objective Function}\label{section3-3}
To optimize the synthesized objective function \eqref{eq19} recursively, the update from sampling time $k-1$ to $k$ is considered. Assume that the latest sample at sampling time $k$ is collected from the partition $\Omega^a$, i.e., $\gamma_a(k)= x(k) \in \Omega^a, 1 \leqslant a \leqslant {M_P}(k)$. Let ${\Omega _{{M_P(k-1)}}} = \mathop  \cup \limits_{j = 1}^{{M_P}(k-1)} \Omega^j  \subset {\Omega}$ be the union set of the partitions which have been sampled at least once before sampling time $k-1$. Then $M_P(k)$ can be calculated by: 
\begin{equation}\label{eq23}
	{M_P}(k) = \left\{ \begin{aligned}
		&{M_P}(k-1),{\text{ }\text{ }\text{ }\text{ }\text{ }\text{ }}for{\text{ }}\forall x(k) \in {\Omega _{{M_P(k-1)}}} \hfill \\
		&{M_P}(k-1) + 1,{\text{   }}for{\text{ }}\forall x(k) \notin {\Omega _{{M_P(k-1)}}}.  \hfill \\ 
	\end{aligned}  \right.
\end{equation}

\begin{theorem}\label{theorem1}
{Let $W_0 \in {\mathbb{R}^N}$, a positive definite matrix $P_0 \in {\mathbb{R}^{N \times N}}$ be the initial values of $W(k)$ and $P(k)$, respectively. Under Assumption \ref{assumption2}, let $W(k)$ denote the minimizer of $J_S(W,k)$ as follows: 
\begin{equation}\label{eq24}
W(k) = \mathop {\arg \min }\limits_{W \in {\mathbb{R}^N}} {J_S}(W,k)
\end{equation}
for any $k>0$ , and then $W(k)$ can be calculated by $W(k-1)$ with the following update laws: 
\begin{equation}\label{eq25}
	\begin{footnotesize}
\begin{aligned}
	&W(k) = W(k-1) + P(k)\Phi ({{\gamma _a}(k)})\left[ {{{\varphi _a}(k)} - {W^T}(k-1)\Phi ({{\gamma _a}(k)})} \right] \hfill \\
	&- P(k)\Phi (\gamma_a(k-1))\left[ {\varphi_a(k-1) - {W^T}(k-1)\Phi (\gamma_a(k-1))} \right] \hfill \\ 
\end{aligned} 
	\end{footnotesize}
\end{equation}
where 
\begin{equation}\label{eq26}
	\begin{small}
		\begin{aligned}
{P^{ - 1}}(k) = {P^{ - 1}}(k - 1) + \Phi ({{\gamma _a}(k)}){\Phi ^T}({{\gamma _a}(k)}) \\
- \Phi (\gamma_a(k-1)){\Phi ^T}(\gamma_a(k-1)). 
		\end{aligned} 
	\end{small}
\end{equation}}
\end{theorem}

\begin{proof}
Since $({\gamma _j}(k),{\varphi _j}(k))$ is the latest sample within $\Omega^j$ before sampling time $k$ (including $k$), the equations 
\begin{equation}\label{eq27}
\left\{ \begin{gathered}
	{\gamma _j}(k) = {\gamma _j}(k - 1) \hfill \\
	{\varphi _j}(k) = {\varphi _j}(k - 1) \hfill \\ 
\end{gathered}  \right.
\end{equation}
hold for any $1 \leqslant j \leqslant {M_P}(k - 1), j \ne a$. 

To find the relationship between $W(k-1)$ and $W(k)$, the following two situations are considered: 

1) If $x(k) \notin {\Omega _{{M_P(k-1)}}}$, we have $a=M_P(k)$, $M_P(k)=M_P(k-1)+1$. Consider \eqref{eq19}, \eqref{eq24} and we obtain: 
\begin{equation}\label{eq28}
\begin{small}
\left\{ \begin{aligned}
	&W(k) = \mathop {\arg \min }\limits_{W \in {\mathbb{R}^N}} {J_S}(W,k) \hfill \\
	&W(k - 1) = \mathop {\arg \min }\limits_{W \in {\mathbb{R}^N}} {J_S}(W,k-1) \hfill \\
	&{J_S}(W,k) = {J_S}(W,k-1) + {\left\| {{\varphi _a}(k) - {W^T}\Phi ({\gamma _a}(k))} \right\|^2}. \hfill \\ 
\end{aligned}  \right.
\end{small}
\end{equation}
According to Lemma \ref{lemma1}, the update from $W(k-1)$ to $W(k)$ can be regarded as a single step of the RLS method. Thus $W(k)$ can be calculated by: 
\begin{equation}\label{eq29}
\begin{footnotesize}
\begin{aligned}
W(k) = W(k - 1) + P(k)\Phi ({\gamma _a}(k))\left[ {{\varphi _a}(k) - {W^T}(k - 1)\Phi ({\gamma _a}(k))} \right]
\end{aligned}
\end{footnotesize}
\end{equation}
where 
\begin{equation}\label{eq30}
\begin{aligned}
	{P^{ - 1}}(k) &= P_0^{ - 1} + \sum\limits_{j = 1}^{{M_P}(k)} {\Phi ({\gamma _j}(k)){\Phi ^T}({\gamma _j}(k))}  \hfill \\ 
	&={P^{ - 1}}(k - 1) + \Phi ({\gamma _a}(k)){\Phi ^T}({\gamma _a}(k)). 
\end{aligned} 
\end{equation}
Since $x(k) \notin {\Omega _{{M_P(k-1)}}}$, we have ${\gamma _a}(k - 1) = {0_{n \times 1}}$ and ${\varphi _a}(k - 1) = 0$. Therefore, the update laws \eqref{eq29}, \eqref{eq30} conform to \eqref{eq25}, \eqref{eq26}, respectively. 

2) If $x(k) \in {\Omega _{{M_P(k-1)}}}$, we have $1 \leqslant a \leqslant {M_P}(k - 1)$, $M_P(k)=M_P(k-1)$.

Define a virtual synthesized objective function as follows: 
\begin{equation}\label{eq31}
{J_V}(W,k) = {J_S}(W,k - 1) - {\left\| {{\varphi _a}(k - 1) - {W^T}\Phi ({\gamma _a}(k - 1))} \right\|^2}. 
\end{equation}
According to \eqref{eq19}, ${J_V}(W,k)$ can also be expressed as: 
\begin{equation}\label{eq32}
{J_V}(W,k) = {J_S}(W,k) - {\left\| {{\varphi _a}(k) - {W^T}\Phi ({\gamma _a}(k))} \right\|^2}. 
\end{equation}
Consider \eqref{eq24}, \eqref{eq31}, \eqref{eq32} and we obtain: 
\begin{equation}\label{eq33}
	\begin{small}
\left\{ \begin{aligned}
	&W(k) = \mathop {\arg \min }\limits_{W \in {\mathbb{R}^N}} {J_S}(W,k) \hfill \\
	&W(k - 1) = \mathop {\arg \min }\limits_{W \in {\mathbb{R}^N}} {J_S}(W,k - 1) \hfill \\
	&{J_S}(W,k) = {J_V}(W,k) + {\left\| {{\varphi _a}(k) - {W^T}\Phi ({\gamma _a}(k))} \right\|^2} \hfill \\
	&{J_S}(W,k - 1) = {J_V}(W,k) + {\left\| {{\varphi _a}(k - 1) - {W^T}\Phi ({\gamma _a}(k - 1))} \right\|^2}. \hfill \\ 
\end{aligned}  \right.
	\end{small}
\end{equation}
Let $W_V(k)$ denote the minimizer of ${J_V}(W,k)$. According to Lemma \ref{lemma1}, since ${J_V}(W,k)$ conforms to the form of \eqref{eq16}, both the update from $W_V(k)$ to $W(k-1)$ and the update from $W_V(k)$ to $W(k)$ can be regarded as a single step of RLS. 

For the update from $W_V(k)$ to $W(k-1)$: 
\begin{equation}\label{eq34}
\begin{footnotesize}
	\begin{aligned}
	&W(k - 1) \\
	= &{W_V}(k) + P(k - 1)\Phi ({\gamma _a}(k - 1))\left[ {{\varphi _a}(k - 1) - W_V^T(k)\Phi ({\gamma _a}(k - 1))} \right]
	\end{aligned}
\end{footnotesize}
\end{equation}
where 
\begin{equation}\label{eq35}
		\begin{aligned}
{P^{ - 1}}(k - 1) =& P_0^{ - 1} + \sum\limits_{j = 1}^{{M_P}(k - 1)} {\Phi ({\gamma _j}(k - 1))} {\Phi ^T}({\gamma _j}(k - 1)) 
		\end{aligned}
\end{equation}

For the update from $W_V(k)$ to $W(k)$: 
\begin{equation}\label{eq36}
		\begin{aligned}
			W(k) = {W_V}(k) + P(k)\Phi ({\gamma _a}(k))\left[ {{\varphi _a}(k) - W_V^T(k)\Phi ({\gamma _a}(k))} \right]
		\end{aligned}
\end{equation}
where 
\begin{equation}\label{eq37}
		\begin{aligned}
			{P^{ - 1}}(k) = P_0^{ - 1} + \sum\limits_{j = 1}^{{M_P}(k)} {\Phi ({\gamma _j}(k))} {\Phi ^T}({\gamma _j}(k))
		\end{aligned}
\end{equation}

Substitute \eqref{eq35} into \eqref{eq37}, and \eqref{eq25} is obtained. 

Consider \eqref{eq25}, and an intermediate variable $P_S(k)$ is introduced as follows: 
\begin{equation}\label{eq38}
	\begin{aligned}
P_S^{ - 1}(k) &= {P^{ - 1}}(k) - \Phi ({\gamma _a}(k)){\Phi ^T}({\gamma _a}(k)) \\
&= {P^{ - 1}}(k - 1) - \Phi ({\gamma _a}(k - 1)){\Phi ^T}({\gamma _a}(k - 1)). 
	\end{aligned}
\end{equation}
Multiply both sides of \eqref{eq34} by ${P^{ - 1}}(k - 1)$ and we obtain: 
\begin{equation}\label{eq39}
	\begin{small}
\begin{aligned}
	{P^{ - 1}}(k - 1)W(k - 1) = P_S^{ - 1}(k){W_V}(k) + \Phi ({\gamma _a}(k - 1)){\varphi _a}(k - 1). \hfill \\ 
\end{aligned} 
	\end{small}
\end{equation}
Multiply both sides of \eqref{eq36} by ${P^{ - 1}}(k)$ and we obtain: 
\begin{equation}\label{eq40}
		\begin{aligned}
			{P^{ - 1}}(k)W(k) = P_S^{ - 1}(k){W_V}(k) + \Phi ({\gamma _a}(k)){\varphi _a}(k). \hfill \\ 
		\end{aligned} 
\end{equation}
Substitute \eqref{eq39} into \eqref{eq40}, and we obtain: 
\begin{equation}\label{eq41}
	\begin{footnotesize}
		\begin{aligned}
		W(k) =& P(k) {P^{ - 1}}(k - 1)W(k - 1) \\
		&+  P(k) \left[ \Phi ({\gamma _a}(k)){\varphi _a}(k) - \Phi ({\gamma _a}(k - 1)){\varphi _a}(k - 1)  \right] \\
		=& W(k-1) + P(k)\Phi ({{\gamma _a}(k)})\left[ {{{\varphi _a}(k)} - {W^T}(k-1)\Phi ({{\gamma _a}(k)})} \right] \hfill \\
		&- P(k)\Phi (\gamma_a(k-1))\left[ {\varphi_a(k-1) - {W^T}(k-1)\Phi (\gamma_a(k-1))} \right] \hfill \\ 
	\end{aligned}
	\end{footnotesize}
\end{equation}
which conforms to \eqref{eq25}. Therefore, the update from $W(k-1)$ to $W(k)$ conforms to \eqref{eq25} and \eqref{eq26} in this case. 

Since the update from $W(k-1)$ to $W(k)$ in both 1) and 2) can be realized using \eqref{eq25}, \eqref{eq26}, the proof is accomplished. 
\end{proof}

According to memory update functions \eqref{eq22}, we have $({\gamma _a}(k),{\varphi _a}(k)) = (x(k),y(k))$, and thus \eqref{eq25}, \eqref{eq26} can be reformulated as: 
\begin{itemize}[leftmargin=*]
	\item \textbf{Selective Memory Recursive Least Squares}
	\begin{equation}\label{eq42}
		\begin{small}
			\begin{aligned}
	&W(k) \\
	=& W(k-1) + P(k)\Phi ({x(k)})\left[ {{y(k)} - {W^T}(k-1)\Phi ({x(k)})} \right] \hfill \\
	&- P(k)\Phi (\gamma_a(k-1))\left[ {\varphi_a(k-1) - {W^T}(k-1)\Phi (\gamma_a(k-1))} \right] \hfill \\ 
&{P^{ - 1}}(k) \\
 =& {P^{ - 1}}(k - 1) + \Phi ({x(k)}){\Phi ^T}({x(k)}) - \Phi (\gamma_a(k-1)){\Phi ^T}(\gamma_a(k-1))
			\end{aligned} 
		\end{small}
	\end{equation}
\end{itemize}
where $({\gamma _a}(k - 1),{\varphi _a}(k - 1))$ represents the latest sample collected from the partition $\Omega^a$ satisfying $x(k) \in {\Omega ^a}$ before sampling time $k-1$ (including $k-1$). Different from classical RLS, SMRLS updates the weights of neural networks with both current approximation errors and recorded data $({\gamma _a}(k - 1),{\varphi _a}(k - 1))$, so extra memory is required. 

\begin{remark}\label{remark4}
{The matrix inversion lemma mentioned in \cite{r39} can also be used to avoid the calculation of $P^{-1}(k)$ in SMRLS. Specifically, let $P_V^{ - 1}(k) = {P^{ - 1}}(k - 1) - \Phi ({\gamma _a}(k - 1)){\Phi ^T}({\gamma _a}(k - 1))$, and it can be calculated with the matrix inversion lemma as: 
\begin{equation}\label{eq43}
	\begin{footnotesize}
	\begin{aligned}
	{P_V}(k) = P(k - 1) + \frac{{P(k - 1)\Phi ({\gamma _a}(k - 1)){\Phi ^T}({\gamma _a}(k - 1))P(k - 1)}}{{1 - {\Phi ^T}({\gamma _a}(k - 1))P(k - 1)\Phi ({\gamma _a}(k - 1))}}. 
\end{aligned}
	\end{footnotesize}
\end{equation}
Then $P(k)$ can also be calculated with the matrix inversion lemma as follows: 
\begin{equation}\label{eq44}
		\begin{aligned}
P(k) = {P_V}(k) - \frac{{{P_V}(k)\Phi (x(k)){\Phi ^T}(x(k)){P_V}(k)}}{{1 + {\Phi ^T}(x(k)){P_V}(k)\Phi (x(k))}}. 
		\end{aligned}
\end{equation}} 
\end{remark}

{An implementation of SMRLS where ${x_i},i = 1,2, \ldots ,n$ represents the normalized input of the RBFNN over $\left[ { - 1,1} \right]$ is given as follows. The normalized input space is evenly divided into $N_P$ lattice partitions. Let ${A=}\left[ {{\alpha_1},{\alpha_2}, \ldots ,{\alpha_{{N_P}}}} \right] \in {\mathbb{R}^{n \times {N_P}}}$, $B = \left[ {{\beta_1},{\beta_2}, \ldots ,{\beta_{{N_P}}}} \right] \in {\mathbb{R}^{1 \times {N_P}}}$ denote the matrix used to save the synthesized sample input ${\gamma _j}(k)$ and output ${\varphi _j}(k),j = 1,2, \ldots ,{N_P}$ of each partition, respectively. To calculate the serial number $a$ of the partition being visited by $x(k)$, i.e. $x(k) \in {\Omega ^a}$, an encoding function is defined with $\left\lceil  \cdot  \right\rceil $ representing round up: 
\begin{equation}\label{eq111}
Q(x(k)) = \sum\limits_{i = 1}^n {\left\lceil {\frac{{({x_i}(k) + 1){{({N_P})}^{\frac{1}{n}}}}}{2} - 1} \right\rceil }  * {({N_P})^{\frac{{i - 1}}{n}}} + 1
\end{equation}
which returns an integer serial number $Q(x(k)) \in \left[ {1,{N_P}} \right]$ for each $x(k)$. Then a specific implementation of SMRLS is shown in Algorithm \ref{SMRLS}. }

\begin{algorithm}
	\SetKwData{Left}{left}\SetKwData{This}{this}\SetKwData{Up}{up}
	\SetKwFunction{Union}{Union}\SetKwFunction{FindCompress}{FindCompress}
	\SetKwInOut{Input}{Input}\SetKwInOut{Output}{Output}
	\textbf{Input: }{sample input $x(k)$, sample output $y(k)$} \\
	\textbf{Output: }{weight vector $W(k)$ of the RBFNN} \\
	\textbf{Initialization: } number of partitions $N_P$, memory space for synthesized samples $A = {0_{n \times {N_P}}}$ and $B = {0_{1 \times {N_P}}}$, initial weight vector ${W_0} \in {\mathbb{R}^N}$, initial positive definite covariance matrix ${P_0} \in {\mathbb{R}^{N \times N}}$ \\
	\For{$k \geqslant 1$}{
		\emph{$a \leftarrow Q(x(k))$}\;
		\emph{${\gamma _a}(k - 1) \leftarrow {\alpha _a}$, ${\varphi _a}(k - 1) \leftarrow {\beta _a}$}\;
			\emph{${\alpha _a} \leftarrow x(k)$, ${\beta _a} \leftarrow y(k)$}\;
		\emph{$P(k) \leftarrow $update law \eqref{eq43}, \eqref{eq44}}\;
		\emph{$W(k) \leftarrow $update law \eqref{eq42}}\;
	}
	\caption{SMRLS}\label{SMRLS}
\end{algorithm}\DecMargin{1em}

\section{Discussions on SMRLS}\label{section4}

\subsection{Computational Complexity}\label{section4-1}
{Since the computational complexity is crucial for a real-time training algorithm, the time complexity of the SGD algorithm \eqref{eq13}, FFRLS algorithm \eqref{eq15}, VDFRLS algorithm mentioned in \cite{r32} and SMRLS algorithm with the covariance matrix update law \eqref{eq43}, \eqref{eq44} is compared. Consider the naive matrix multiplication, and the time complexity of these algorithms is shown in TABLE \ref{table1} with $N$ being the number of neurons. As a result of the matrix inversion lemma, the inversion operation in FFRLS and SMRLS is avoided reducing the time complexity to $O(N^2)$, while the singular value decomposition in VDFRLS cannot be avoided leading to the time complexity of $O(N^3)$. Compare the update laws of FFRLS and SMRLS, and it is shown that the time complexity of SMRLS is at most twice over FFRLS indicating the feasibility of SMRLS in practical real-time learning scenarios. }

\begin{table}[H]
	\caption{{Time Complexity}}
	\centering
	\label{table1}
	\begin{tabular}{ccccc}
		\toprule
		 Algorithm&SGD &FFRLS &VDFRLS &SMRLS \\
		\midrule 
		Time Complexity&$O(N)$&$O(N^2)$&$O(N^3)$&$O(N^2)$ \\
		\bottomrule
	\end{tabular}
\end{table}

{
The implementation of SMRLS requires each partition divided from the normalized input space to store the latest sample within it, so the space complexity of SMRLS is $O(N_P)$ with $N_P$ being the number of partitions. }

{
\begin{remark}\label{remark5}
The number of neurons and partitions of the input space will dramatically increase with the increase of the input space dimension, which leads to increased time and space complexity of SMRLS. Therefore, when the input dimension of the neural network is large, various feature selection and pruning methods such as the principal component analysis (PCA) and sensitivity analysis can be used to reduce the computational complexity of SMRLS \cite{r40, r41, r42, r43, r44}. 
\end{remark}}

\subsection{Learning Performance}\label{section4-2}
Compared with classical real time training methods with forgetting mechanisms, the learning performance (learning speed, accuracy and generalization capability) of SMRLS is given as follows: 
\begin{enumerate}
\item{The learning speed of the RBFNN is improved because the passive knowledge forgetting phenomenon is suppressed by SMRLS. With the memory mechanism, the importance of a sample is determined not only by when it is collected but also by its position in the input space. Therefore, the useful knowledge will not get lost simply because they are learned a long time ago, and thus the learning speed is improved. 
}
\item{The generalization capability of the learned knowledge is improved with SMRLS. As analyzed in Remark \ref{remark3}, the synthesized samples are approximately uniformly distributed over the normalized input space, which suppresses the non-uniformity between the training set and testing set. 
}
\item{{The learning accuracy of SMRLS is dependent on the density of the partitions. Although the partitioning operation helps to establish the memory mechanism, it also determines the maximal number of samples. Since too few samples will lead to unsatisfactory learning accuracy, $N_P$ should be designed properly. Compared with training methods with forgetting mechanisms, SMRLS generally achieves lower accuracy in regions where the samples are dense and higher accuracy in regions where the samples are sparse because of its approximately uniformly distributed synthesized samples. }
}
\end{enumerate}
\subsection{Prospects and Possible Improvements}\label{section4-3}
With its memory mechanism, SMRLS is suitable for long-term real-time learning tasks of the RBFNN where the accumulation of knowledge is important. In addition to the training of RBFNNs, SMRLS can also be used to update the parameters of other linearly parameterized approximators such as polynomials, autoregressive models and single layer feedforward networks \cite{r11,r39,r45,r46}. Therefore, it is promising in the field of neural network control, system identification and machine learning. 

Meanwhile, SMRLS still has some issues to be investigated including: i) other possible designs of the memory update functions, ii) different designs of the partitioning operation, iii) methods to reduce the computational complexity.

\section{Simulation Studies}\label{section5}
In this section, an RBFNN is employed to approximate an unknown function $f_I(x)$ of a single inverted pendulum-cart system, whose state space description is formulated as follows: 
\begin{equation}\label{eq45}
	\left\{ \begin{aligned}
		{{\dot x}_1} &= {x_2} \hfill \\
		{{\dot x}_2} &= {f_I}(x) + {g_I}(x)u \hfill \\
	\end{aligned}  \right.
\end{equation}
where 
\begin{equation}\label{eq46}
	{f_I}(x) = \frac{{g\sin {x_1} - mlx_2^2\cos {x_1}\sin {x_1}/({m_c} + m)}}{{l(4/3 - m{{\cos }^2}{x_1}/({m_c} + m))}}. 
\end{equation}
The state vector $x = {\left[ {{x_1},{x_2}} \right]^T} \in \Omega_x$, where $x_1$ and $x_2$ represent the angular position and velocity of the pendulum, respectively, $g=9.8m/s^2$ is the gravitational acceleration, $m_c=0.1kg$ is the weight of the cart, $m=0.02kg$ is the weight of the pendulum, and $l=0.2m$ is half the length of the pendulum. It is assumed that the output of the unknown function $f_I(x)$ can be measured directly as follows: 
\begin{equation}\label{eq47}
	y(k) = {f_I}(x(k)). 
\end{equation}

{
According to \eqref{eq18}, the input space $\Omega_x$ is normalized into the scope of $\left[ { - 1,1} \right] \times \left[ { - 1,1} \right]$, with $\bar x = {\left[ {{{\bar x}_1},{{\bar x}_2}} \right]^T}$ being the normalized state vector. Let $\bar x$ be the input of the RBFNN, and the objective of the real-time approximation is to approximate $f_I(x)$ with ${W^T}\Phi (\bar x)$ over the input space $\Omega_x$ such that the following objective function is minimized: 
\begin{equation}\label{eq48}
J(W) = \int_{{\Omega _x}} {{{\left\| {{f_I}(x) - {W^T}\Phi (\bar x))} \right\|}^2}} dx. 
\end{equation}
For hyperparameter settings of the RBFNN, there are $3 \times 3$ neurons evenly distributed over $\left[ { - 1,1} \right] \times \left[ { - 1,1} \right]$ with the receptive field width of ${\sigma _i} = 1$. }

{In this section, the following real-time training methods are compared: SGD \eqref{eq13}, FFRLS \eqref{eq15}, VDFRLS mentioned in \cite{r32}, and SMRLS \eqref{eq42} with the covariance matrix calculation mentioned in Remark \ref{remark4}. For fairness of the comparisons, the selection of parameters follows two rules: the parameter $P_0$ shared by all RLS based methods are set to the same value that yields a good performance, and the other parameters are optimized by trial and error. The hyperparameters of these methods are shown in TABLE \ref{table2} with $I$ being the identity matrix. The forgetting threshold of VDFRLS is set to $\varepsilon=0.5$. When SMRLS is used to train the RBFNN, the normalized input space $\left[ { - 1,1} \right] \times \left[ { - 1,1} \right]$ is evenly divided into $100 \times 100$ partitions. The sampling period (simulation step size) is set to $0.01s$. }

\begin{table}[h]
	\centering
	\caption{{Hyperparameters of Different Methods}}
	\label{table2}
	\begin{tabular}{|c|c|c|c|c|}
		\hline
		&$W_0$&$P_0$&$\eta$&$\lambda$ \\
		\hline
		SMRLS&$0$&$10I$&$-$&$-$ \\
		\hline
		SGD&$0$&$-$&$0.02$&$-$ \\
		\hline
		FFRLS&$0$&$10I$&$-$&$0.999$ \\
		\hline
		VDFRLS&$0$&$10I$&$-$&$0.999$ \\
		\hline
	\end{tabular}
\end{table}

\subsection{Simulation Studies for Sensitivity to New Data}\label{section5-1}
To test the sensitivity of different methods to new data, the following real-time learning scenarios are considered. 

\noindent \textbf{(A) Repetitive state trajectory with abrupt parameter perturbation}

Consider the following state trajectory of ${x_1}(t)$ and corresponding measurement: 
\begin{equation}\label{eq49}
\left\{ \begin{gathered}
	{x_1}(t) = \sin t \hfill \\
	y(k) = {f_I}(x(k)) \hfill \\ 
\end{gathered}  \right.
\end{equation}
where the value of $l$ in \eqref{eq46} is changed abruptly as below: 
\begin{equation}\label{eq50}
l = \left\{ \begin{gathered}
	0.2, \text{ } for \text{ } 0 \leqslant t < 50s \hfill \\
	0.3, \text{ } for \text{ } 50s \leqslant t \leqslant 100s.  \hfill \\ 
\end{gathered}  \right.
\end{equation}
The training of the RBFNN is carried out over $\left[ {0,100s} \right]$. The input space $\Omega_x$ is set to the state trajectory $\mu(x(0))$, i.e. the RBFNN is expected to realize accurate approximation of $f_I(x)$ along $\mu(x(0))$. 

Fig. \ref{fig2} and Fig. \ref{fig3} demonstrate the convergence process of the approximation errors and weights, respectively. In the first $50s$ of the training, the RBFNNs have approximated the unknown function $f_I(x)$ with the parameter $l=0.2m$, and thus when the parameter abruptly changes, the RBFNN needs to replace the learned data with new data. Compared with other methods, SMRLS achieves higher learning (weight convergence) speed after the parameter perturbation. According to the derivation of SMRLS in Section \ref{section3}, If the sampling period is small enough, SMRLS can replace the outdated data in one period for any repetitive trajectories. 
\begin{figure*}
	\makebox[\textwidth][c]{\includegraphics[height=0.2\textwidth,width=1\textwidth]{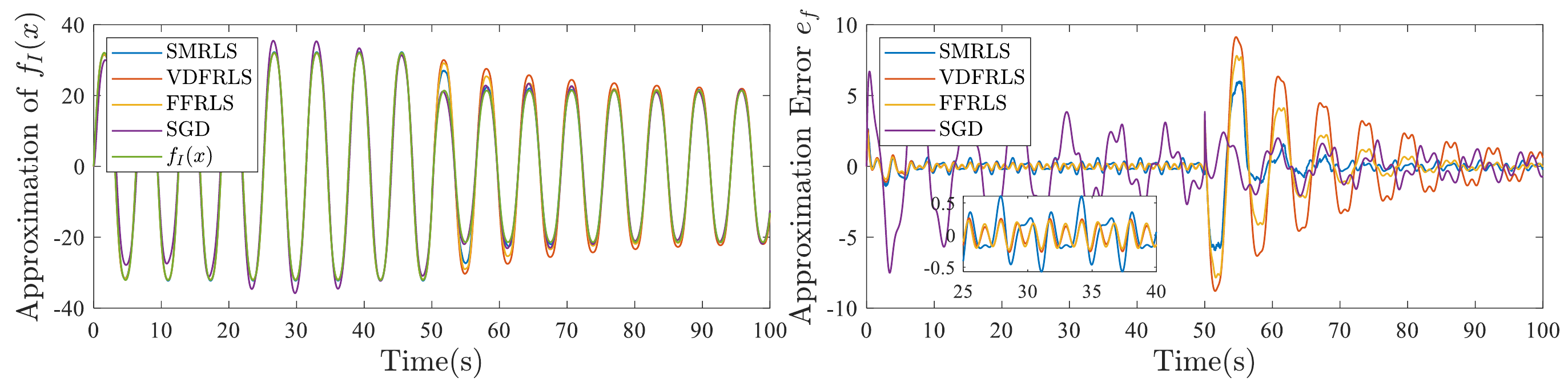}}
	\caption{Real-time tracking performance in case (A) }
	\label{fig2}\vspace{-1em}
\end{figure*}

\begin{figure*}
	\makebox[\textwidth][c]{\includegraphics[height=0.38\textwidth,width=1.03\textwidth]{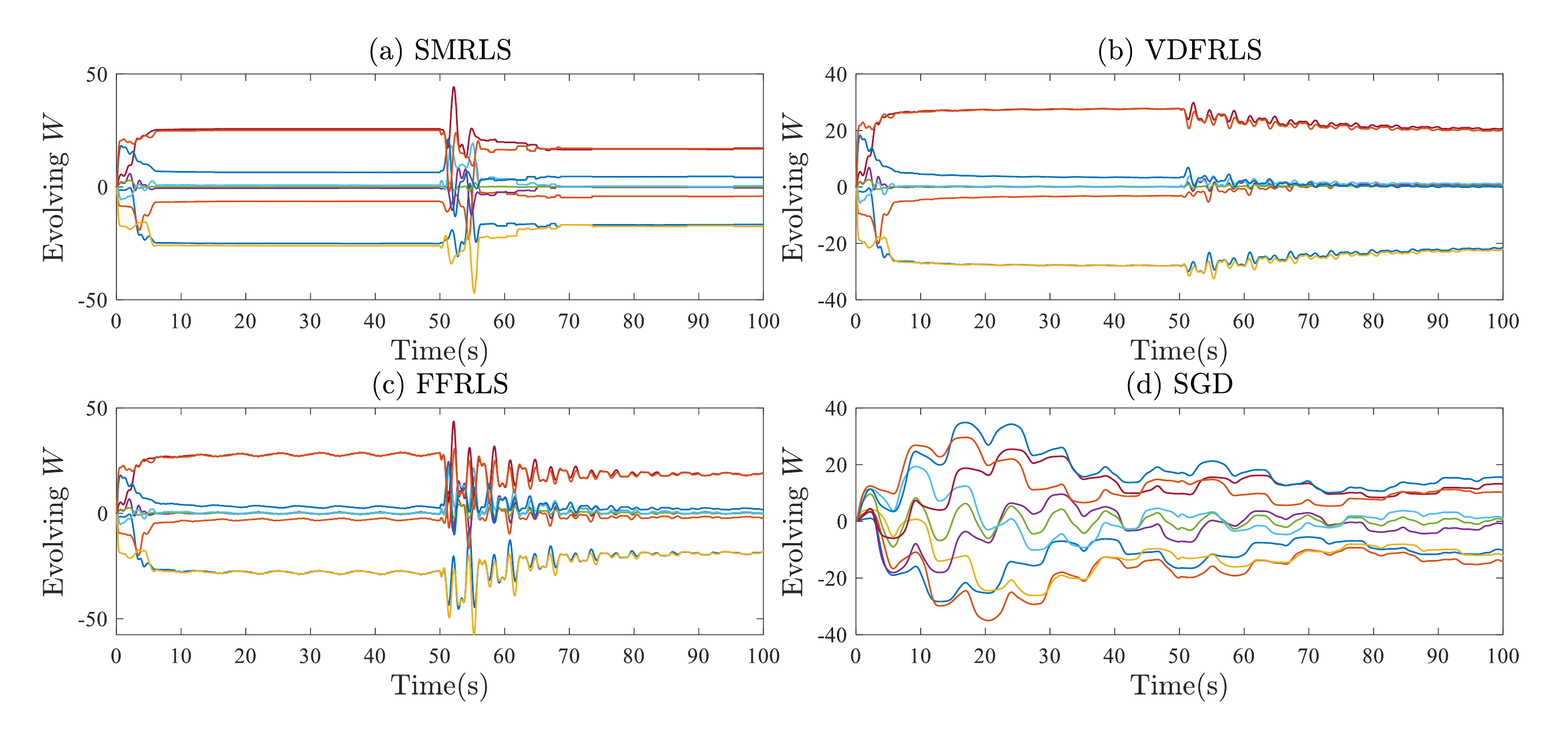}}
	\caption{Convergence of the weights in case (A) }
	\label{fig3}\vspace{-1em}
\end{figure*}

To test the learning accuracy, the weights at $100s$ are recorded as the learned knowledge to approximate $f_I(x)$ with the latest parameter $l=0.3m$. Let ${e_f} = {f_I}(x) - {W^T}\Phi (\bar x)$ denote the approximation error. The approximation using the learned knowledge along \eqref{eq49} is shown in Fig. \ref{fig4}. While all of the four methods achieve the basic approximation of $f_I(x)$, the accuracy of the RLS based methods are obviously better than SGD because the passive knowledge forgetting greatly reduces the learning speed of SGD. Among the RLS based methods, FFRLS achieves the highest accuracy because the forgetting mechanism is suitable for persistently exciting input signals \cite{r7,r32}. The reason why SMRLS-R is less accurate than FFRLS lies in the limitation of the maximum number of samples, i.e. $N_P$, mentioned in Section \ref{section4-2}. The increase of $N_P$ will improve the approximation accuracy but also increases the computational complexity and the requirement for sampling period.

\begin{figure*}
	\makebox[\textwidth][c]{\includegraphics[height=0.19\textwidth,width=1\textwidth]{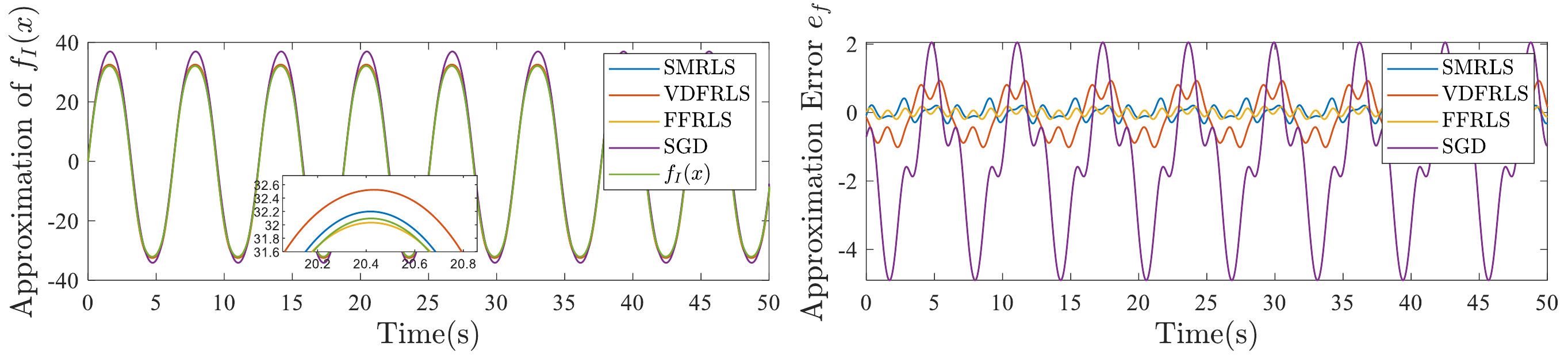}}
	\caption{Approximation using learned knowledge in case (A) }
	\label{fig4}\vspace{-1em}
\end{figure*}

\subsection{Simulation Studies for Passive Knowledge Forgetting}\label{section5-2}
\noindent \textbf{(B) Non-repetitive state trajectory}

To test the suppression effect of SMRLS to passive knowledge forgetting, the following non-repetitive state trajectory is considered: 
\begin{equation}\label{eq51}
\left\{ \begin{gathered}
	{x_1}(t) = \frac{{(20 + t)\sin t}}{{120}} \hfill \\
	y(k) = {f_I}(x(k)) \hfill \\ 
\end{gathered}  \right.
\end{equation}
where the duration of the learning task is set to $100s$. The input space $\Omega_x$ is also set to the state trajectory $\mu(x(0))$. 

The spiral line trajectory \eqref{eq51} is shown in Fig. \ref{fig5}. Fig. \ref{fig6} shows the continuously updated weights as increasing knowledge is learned by the RBFNN. Fig. \ref{fig7} shows the real-time tracking effect of different methods, where FFRLS has the highest tracking accuracy. However, high tracking accuracy does not mean high learning accuracy. 

\begin{figure}[htbp]
	\centering
	\includegraphics[height=0.19\textwidth,width=0.5\textwidth]{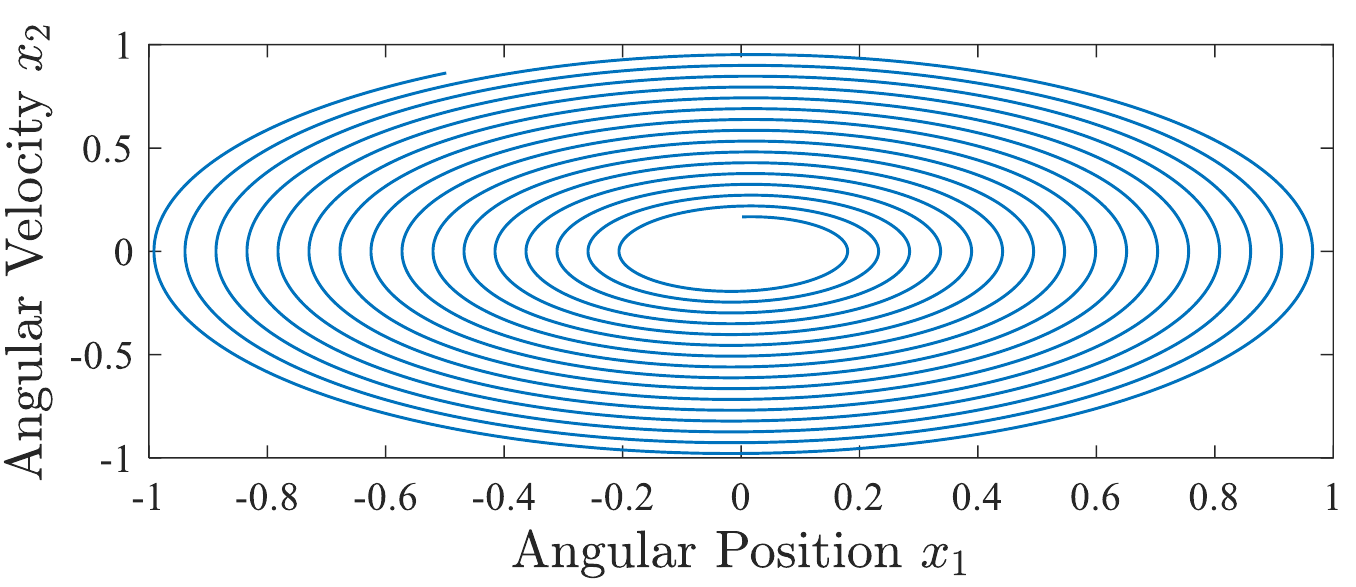}
	\caption{Non-repetitive trajectory of case (B)}\vspace{-1em}\label{fig5}
\end{figure}

\begin{figure}[htbp]
	\centering
	\includegraphics[height=0.19\textwidth,width=0.5\textwidth]{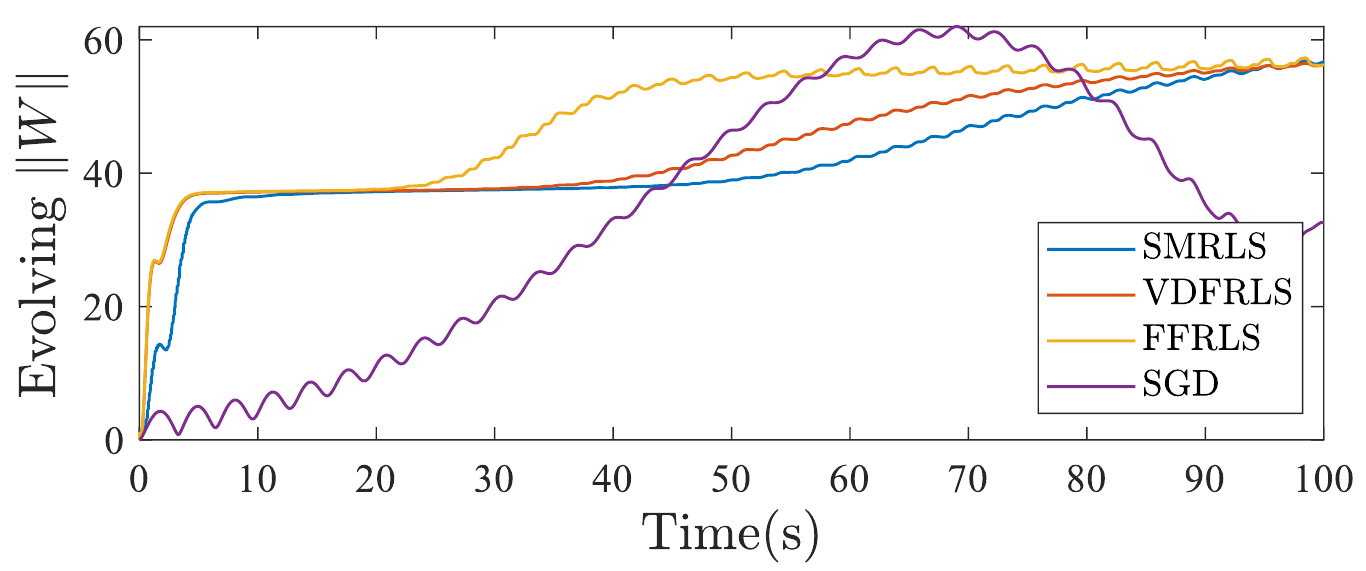}
	\caption{Evolving weight vectors of RBFNNs in case (B)}\vspace{-1em}\label{fig6}
\end{figure}

\begin{figure*}
	\makebox[\textwidth][c]{\includegraphics[height=0.19\textwidth,width=1\textwidth]{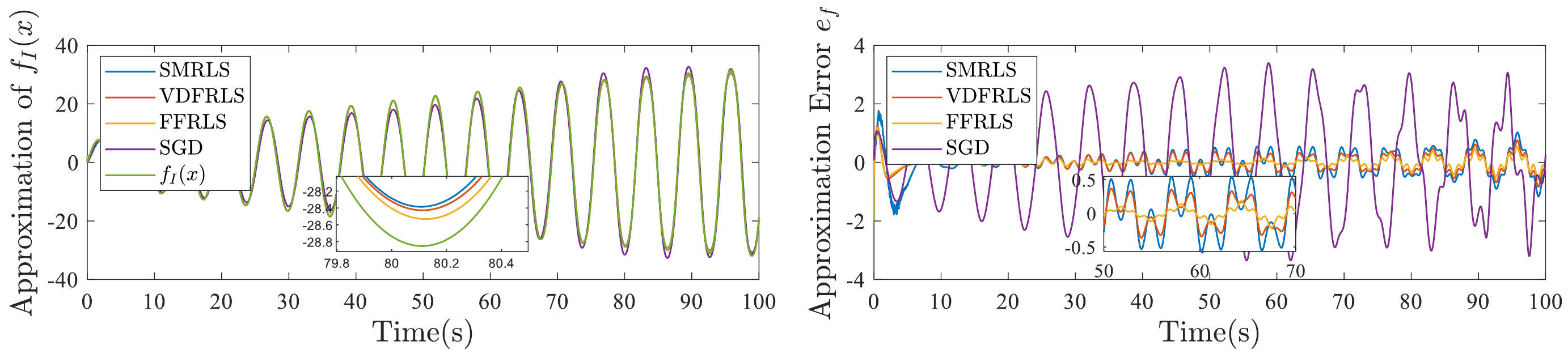}}
	\caption{Real-time tracking performance in case (B)}
	\label{fig7}\vspace{-1em}
\end{figure*}

The weight vectors at $100s$ are recorded as the learned knowledge and Fig. \ref{fig10} demonstrates the learning accuracy of the methods. The forgetting mechanism of FFRLS improves its real-time tracking accuracy while reducing the accuracy of the learned knowledge because the past samples are gradually forgotten. Therefore, the learned knowledge of FFRLS is less accurate over $\left[ {0,90s} \right]$ compared with SMRLS and VDFRLS. It is also shown that the passive knowledge forgetting phenomenon is suppressed with SMRLS and VDFRLS. 

\begin{figure*}
	\makebox[\textwidth][c]{\includegraphics[height=0.19\textwidth,width=1\textwidth]{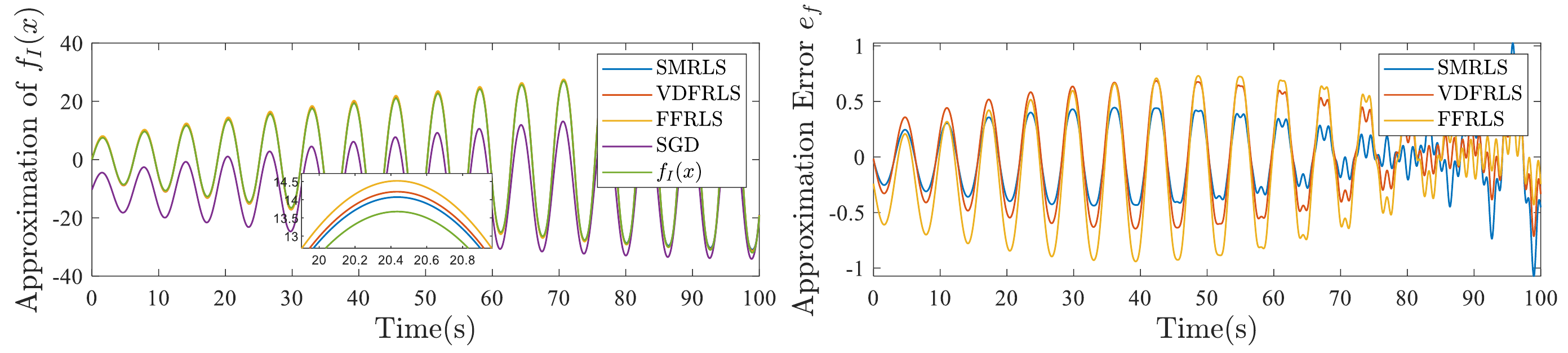}}
	\caption{Approximation using learned knowledge in case (B)}
	\label{fig8}\vspace{-1em}
\end{figure*}

\subsection{Simulation Studies for Generalization Capability}\label{section5-3}
\noindent \textbf{(C) Non-uniformly distributed state trajectory}

One of the most important merits of SMRLS is its good generalization capability obtained from its approximately uniformly distributed synthesized samples. To demonstrate this merit, the random non-uniform rational B-splines (NURBS) trajectory $\mu (x(0))$ in Fig. \ref{fig9} is adopted. In this case, the input space $\Omega_x$ is set to $\left[ { - 1,1} \right] \times \left[ { - 1,1} \right]$, which means that the RBFNN is expected to achieve a globally accurate approximation of $f_I(x)$ over $\left[ { - 1,1} \right] \times \left[ { - 1,1} \right]$. As shown in case (A) and (B), SGD performs worse than the RLS based methods, so its corresponding results are omitted in this case. 

\begin{figure*}
	\makebox[\textwidth][c]{\includegraphics[height=0.19\textwidth,width=1\textwidth]{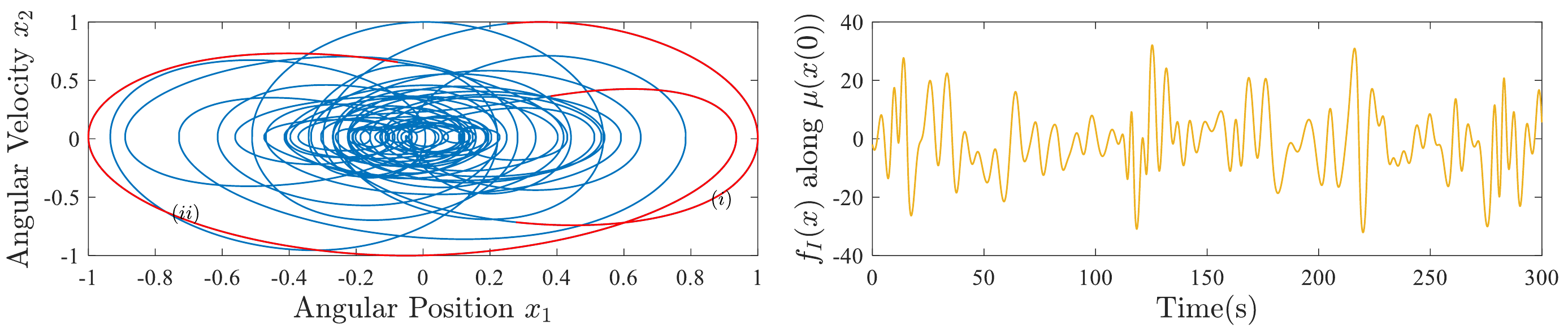}}
	\caption{Random NURBS trajectory of case (C)}
	\label{fig9}\vspace{-1em}
\end{figure*}

Fig. \ref{fig10} demonstrates the approximation errors and evolving weights in the learning process, indicating that FFRLS has higher real-time tracking accuracy while SMRLS and VDFRLS achieve faster learning. The weight vectors at $300s$ are recorded as the learned knowledge and Fig. \ref{fig11} demonstrates the learning accuracy of these methods. The two subfigures $(i), (ii)$ in Fig. \ref{fig10} and Fig. \ref{fig11} are corresponding to the two segments in red of Fig. \ref{fig9}, i.e. the segments whose surrounding samples are sparse. As shown in Fig. \ref{fig11}, compared with FFRLS and VDFRLS, SMRLS achieves more balanced approximation accuracy along the trajectory, i.e. higher accuracy over the regions whose samples are sparse and lower accuracy over the regions whose samples and dense compared with FFRLS. 

\begin{figure*}
	\makebox[\textwidth][c]{\includegraphics[height=0.19\textwidth,width=1\textwidth]{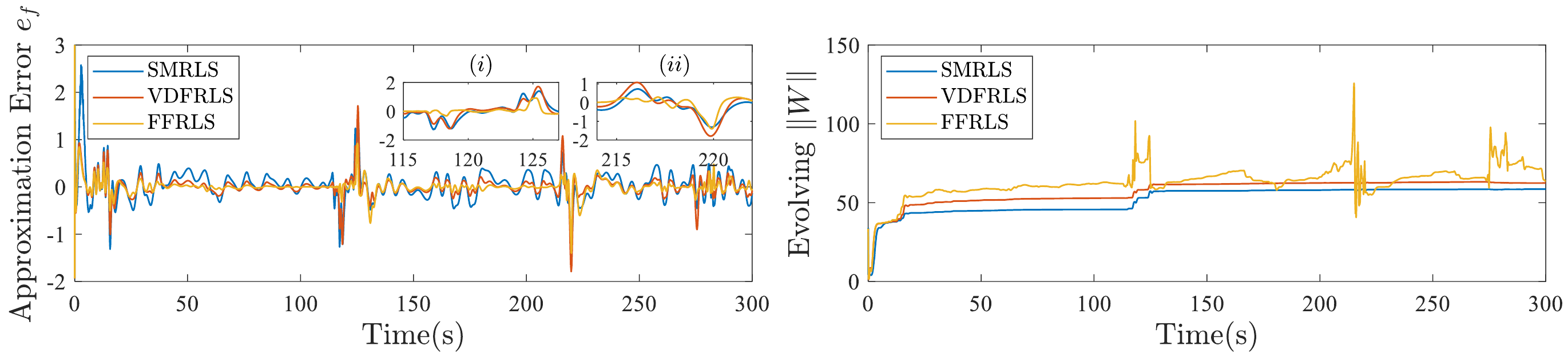}}
	\caption{Tracking performance and weight convergence in case (C)}
	\label{fig10}\vspace{-1em}
\end{figure*}

\begin{figure}[htbp]
	\centering
	\includegraphics[height=0.19\textwidth,width=0.5\textwidth]{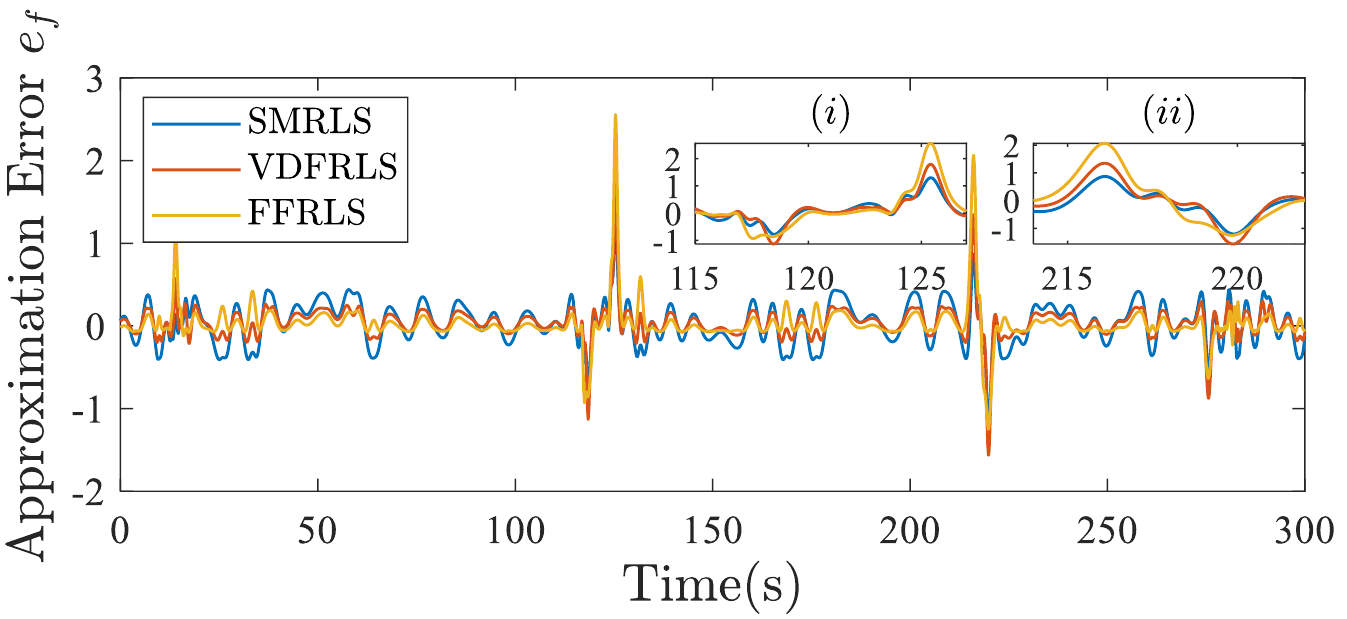}
	\caption{Approximation using learned knowledge in case (C)}\vspace{-1em}\label{fig11}
\end{figure}

{To test the generalization capability of the learned knowledge, the recorded weight vector is adopted to approximate $f_I(x)$ along trajectory \eqref{eq51} over $\left[ {0,100s} \right]$, which nearly traverses the input space $\left[ { - 1,1} \right] \times \left[ { - 1,1} \right]$. As shown in Fig. \ref{fig12}, the knowledge learned with SMRLS achieves more balanced approximation accuracy over $\left[ { - 1,1} \right] \times \left[ { - 1,1} \right]$ and thus obtains better generalization capability between different tasks. }

\begin{figure}[htbp]
	\centering
	\includegraphics[height=0.19\textwidth,width=0.5\textwidth]{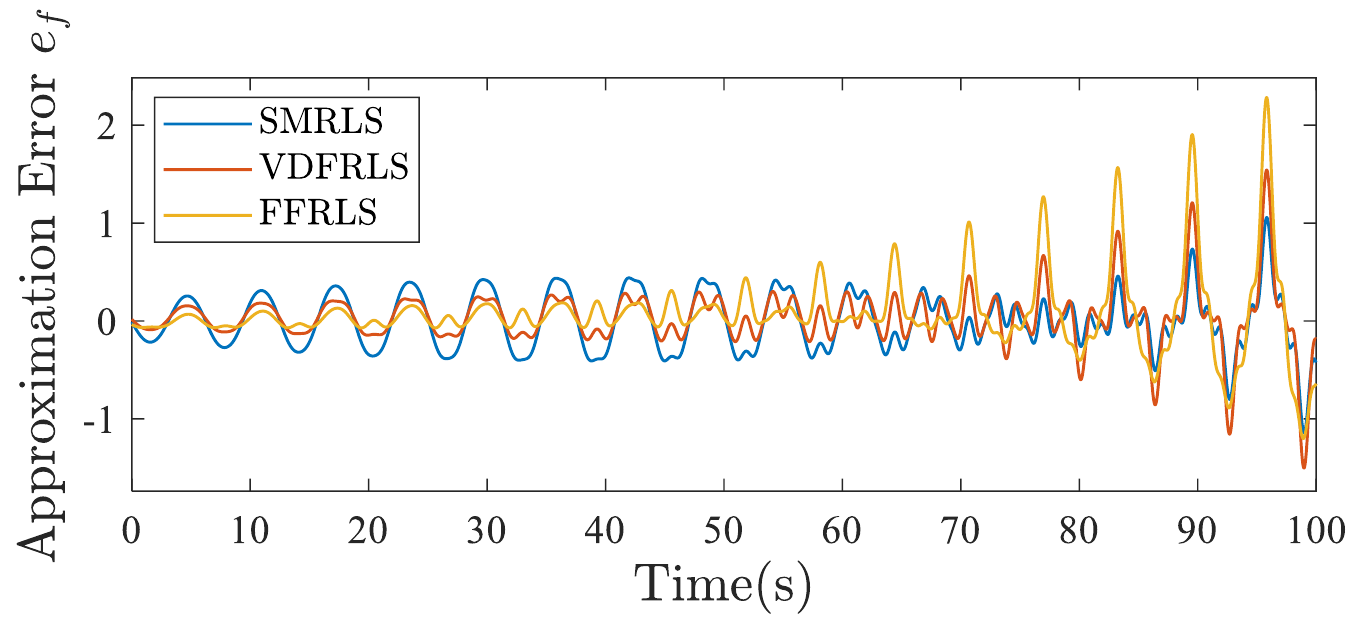}
	\caption{Approximation using learned knowledge along an ergodic trajectory}\vspace{-1em}\label{fig12}
\end{figure}

\subsection{Simulation Studies for Global RBFNN}
{In addition to localized approximation, the RBFNN also has the global approximation capability as mentioned in Section \ref{section2-2}. To test the effectiveness of SMRLS when the global RBFNN is adopted, the simulation of case (B) is repeated with a larger receptive field width ${\sigma _i} = 2$ while other settings remain unchanged. Fig. \ref{fig13} and Fig. \ref{fig14} demonstrate the tracking effect and weight convergence in the learning process. Record the weights at $100s$ as the learned knowledge and Fig. \ref{fig15} shows the approximation using the learned knowledge along \eqref{eq51}. Although the tracking and learning accuracy of the RLS based methods is lower compared with the results in Fig. \ref{fig7} and Fig. \ref{fig8}, the global RBFNN also achieves relatively accurate approximation of $f_I(x)$. As shown in Fig. \ref{fig15}, SMRLS can also suppress passive knowledge forgetting in this case. } 

\begin{figure*}
	\makebox[\textwidth][c]{\includegraphics[height=0.19\textwidth,width=1\textwidth]{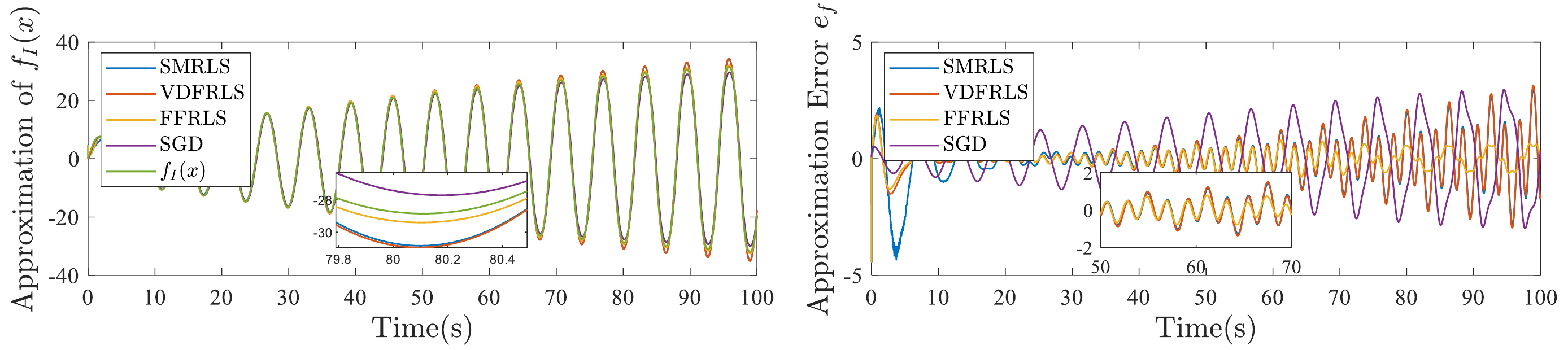}}
	\caption{Tracking performance with the global RBFNN}
	\label{fig13}\vspace{-1em}
\end{figure*}

\begin{figure}[htbp]
	\centering
	\includegraphics[height=0.19\textwidth,width=0.5\textwidth]{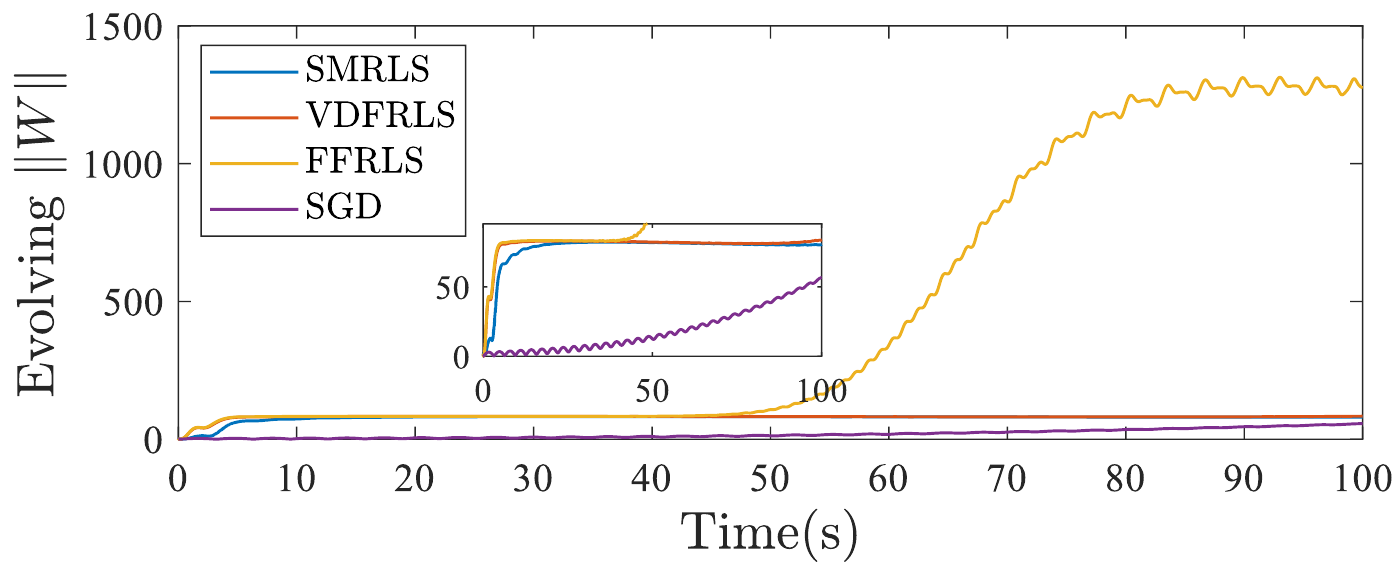}
	\caption{Weight convergence of the global RBFNN}\vspace{-1em}\label{fig14}
\end{figure}

\begin{figure}[htbp]
	\centering
	\includegraphics[height=0.38\textwidth,width=0.5\textwidth]{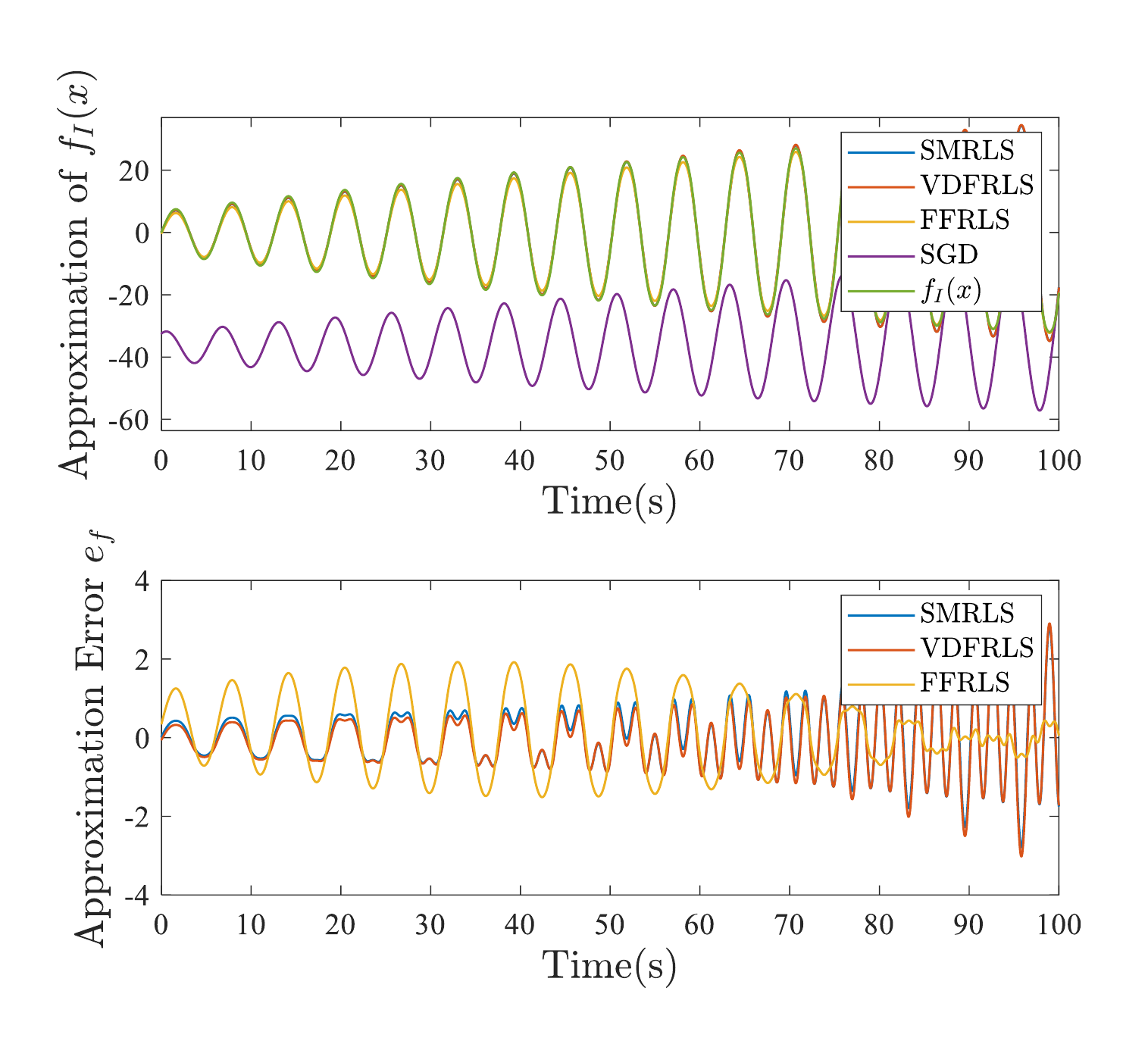}
	\caption{Approximation using learned knowledge with the global RBFNN}\vspace{-1em}\label{fig15}
\end{figure}

The aforementioned simulation results demonstrate that SMRLS achieves faster learning speed and better generalization capability compared with FFRLS and VDFRLS while the accuracy of it is determined by the number of partitions. {It is important to note that the learning accuracy of SMRLS cannot always be improved by increasing the number of partitions, as this will lead to greater computational complexity and reduce the sensitivity of the neural network to new data. }

\section{Conclusion}\label{section6}
In this paper, an RBFNN based real-time function approximation problem is studied and an algorithm named selective memory recursive least squares is proposed to suppress the passive knowledge forgetting phenomenon caused by classical forgetting mechanisms. Featured with the memory mechanism, SMRLS evaluates the importance of samples through both temporal and spatial distribution of the samples such that useful knowledge will not be forgot by the objective function simply because it is learned a long time ago. Compared with training methods with forgetting, SMRLS achieves the merits of fast learning speed and good generalization capability. In addition to the RBFNN, SMRLS is also applicable for real-time training of other linearly parameterized approximators such as polynomials, autoregressive models and other single layer feedforward networks. Moreover, as a result of its good learning performance and acceptable computational complexity, SMRLS has considerable potentials in the field of adaptive control, system identification and machine learning.


\ifCLASSOPTIONcaptionsoff
  \newpage
\fi

\small
\bibliographystyle{ieeetr}
\bibliography{refer}

\end{document}